\documentclass[a4paper,UKenglish,cleveref, autoref]{lipics-v2019}

\usepackage[utf8]{inputenc}
\usepackage{hyperref}
\usepackage{stmaryrd}
\usepackage{amsmath,amssymb}
\usepackage{cancel}
\usepackage{xspace}
\usepackage{comment}
\usepackage{tikz}
\usepackage{xcolor}

\newcommand{\set}[1]{\left\{#1 \right\}}
\newcommand{\tuple}[1]{\left(#1 \right)}
\newcommand{\sem}[1]{ \llbracket #1 \rrbracket}
\newcommand{\dist}[1]{ \lVert #1 \rVert}

\newcommand{\infix}[2]{(#1{:}#2)}
\newcommand{\pref}[1]{\infix{}{#1}}
\newcommand{\suf}[1]{\infix{#1}{}}

\newcommand{\ptime}{\textsc{PTime}\xspace}

\newcommand{\nat}{\mathbb N}

\newcommand{\A}{\mathcal A}
\newcommand{\R}{\mathcal R}
\newcommand{\F}{\mathcal F}
\newcommand{\Sy}{\mathcal S}
\newcommand{\T}{\mathcal T}
\newcommand{\M}{\mathcal M}
\newcommand{\Pp}{\mathcal P}
\newcommand{\final}{\mathit t}
\newcommand{\id}{Id}
\newcommand{\dom}{\mathrm{dom}}
\newcommand{\ran}{\mathrm{ran}}
\newcommand{\fin}{\mathit{t}}
\newcommand{\finite}{\mathit{ fin}}

\newcommand{\rateq}{{\bf RatEq}\xspace}
\newcommand{\kerrat}{{\bf KerRat}\xspace}
\newcommand{\kersseq}{{\bf KerSub}\xspace}
\newcommand{\kerseq}{{\bf KerSeq }\xspace}

\newcommand{\lp}{{lp}}
\newcommand{\ltl}{{ll}}

\newcommand{\ie}{\textit{i.e.}~}

\title{Equivalence kernels of sequential functions
 and sequential observation synthesis} 

\titlerunning{Kernels of sequential functions}

\author{Paulin Fournier}{Université de Nantes, France}{}{}{}

\author{Nathan Lhote}{University of Warsaw, Poland}{}{}{The second author was supported by the french National Reasearch Agency (ANR) DeLTA project (ANR-16-CE40-0007) and by the European Research Council (ERC) grant
 under the European Union’s Horizon 2020 research and innovation programme (ERC Consolidator
 Grant LIPA, grant agreement No. 683080).}

\authorrunning{P. Fournier and N. Lhote}

\Copyright{Paulin Fournier and Nathan Lhote}



\category{}

\relatedversion{}

\supplement{}

\acknowledgements{}

\nolinenumbers 

\EventEditors{John Q. Open and Joan R. Access}
\EventNoEds{2}
\EventLongTitle{42nd Conference on Very Important Topics (CVIT 2016)}
\EventShortTitle{CVIT 2016}
\EventAcronym{CVIT}
\EventYear{2016}
\EventDate{December 24--27, 2016}
\EventLocation{Little Whinging, United Kingdom}
\EventLogo{}
\SeriesVolume{42}
\ArticleNo{23}

\begin{document}

\maketitle

\begin{abstract}
We show that one can decide if a rational equivalence relation can be given as the equivalence kernel of a sequential letter-to-letter transduction. This problem comes from the setting of games with imperfect information. In \cite[p.~6]{BerwangerD18} the authors propose to model imperfect information by a rational equivalence relation and leave open the problem of deciding if one can synthesize a sequential letter-to-letter transducer (Mealy machine) which maps equivalent histories to the same sequence of observations. We also show that knowing if an equivalence relation can be given as the equivalence kernel of a sequential transducer is \emph{undecidable}, even if the relation is given as a letter-to-letter transducer.
\end{abstract}

\newpage
\section*{Introduction} 

\subparagraph{Motivation: games with imperfect information}
The motivation for the present article comes from the paper: \emph{Observation and Distinction. Representing Information in Infinite Games} by Dietmar Berwanger and Laurent Doyen, submitted to the arXiv in 2018 \cite{BerwangerD18}.
The authors propose an alternative way of representing imperfect information in games. The standard way to model imperfect information for a player is through a Mealy machine which transforms a sequence of game locations (a history) into a sequence of observations, which we call in the following an \emph{observation function}.
The proposed model of \cite[p.~6]{BerwangerD18} is to give instead a transducer recognizing an \emph{indistinguishability relation}, \textit{i.e.}~an equivalence relation over game histories which recognizes those pairs of histories that are indistinguishable from the player's perspective.

This new model is actually more expressive than the standard one (composing a Mealy machine with its inverse yields a transducer recognizing the indistinguishability relation), and one of the problems left open in \cite[p.~22]{BerwangerD18} is to decide when an indistinguishability relation can be transformed into an observation function, given as a Mealy machine.

Given a class \textbf R of equivalence relations and a class \textbf F of functions we define the \textbf R,\textbf F-\emph{observation synthesis problem} as the problem of deciding if an equivalence relation in \textbf R can be expressed as the equivalence kernel\footnote{The equivalence kernel of a total function $f$ is defined by $x\sim y \Leftrightarrow f(x)=f(y)$} of a function in \textbf F, and if possible computing such a function.

The main goal of this article is to solve this problem for rational relations and functions given by Mealy machines. Moreover, we also consider the problem of constructing an observation function given, not as a Mealy machine but, as a sequential transducer, \textit{i.e.}~the outputs are not restricted to single letters but can be arbitrary words.
In terms of observations, Mealy machines characterize the fact that each game move produces exactly one piece of observation (in some finite alphabet), while for sequential transducers, a move might produce several observations, or even none, in which case this step is \emph{invisible} to the player.

\subparagraph{Contributions}
We don't use the vocabulary of games, but that of transducers, which is actually more suited to this problem: most of the proof techniques that we use stem from the theory of transducers.
We consider several subclasses of \rateq, the set of rational equivalence relations, that is relations realized by transducers.
The \emph{equivalence kernel} of a total function $f$, is the equivalence relation defined by having the same image under $f$.
The class $\kerseq$ contains the equivalence relations that are the equivalence kernels of \emph{sequential} transductions (a transducer is \emph{sequential} if it is deterministic with respect to the input).
The subclass $\kerseq^\ltl$ is the set of equivalence relations that are the equivalence kernel of a transduction given as a \emph{sequential letter-to-letter} transducer (also known as a Mealy machine).

We start by studying the simpler class of $\kerseq^\ltl$ in Sec.~\ref{sec:ltl} and then consider the class $\kerseq$ in Sec.~\ref{sec:seq}.
Our main contribution is to give explicit characterizations for both classes $\kerseq$ and $\kerseq^\ltl$.
For relations satisfying these properties, we exhibit a construction of a sequential, resp. letter-to-letter sequential, transducer whose kernel is the original relation.
Finally we show that for rational equivalence relations, membership in $\kerseq^\ltl$ is \emph{decidable}.
In contrast, membership in $\kerseq$ is \emph{undecidable} even for letter-to-letter rational relations (also known as automatic, synchronous or regular relations).

Note that while the characterization of $\kerseq^\ltl$, as well as the construction were already given in \cite[Thm.~29, p.~19]{BerwangerD18}, the decidability status was left open. We reprove these results in our framework.
Moreover, while extending the construction from $\kerseq^\ltl$ to $\kerseq^\lp$ is rather straightforward, obtaining the characterization for this class is difficult and actually the \emph{most} challenging part of this article.

\section{Words, relations, automata and transducers}

\subparagraph{Words, languages and relations}

An \emph{alphabet} $A$ is a set of symbols called \emph{letter}. A word is a finite sequence of letters and we denote by $A^*$ the set of finite words with $\epsilon$ denoting the \emph{empty word}. The length of a word $w$ is denoted by $|w|$ with $|\epsilon|=0$. Given a non-empty word $w$ and an integer $1\leq i\leq |w|$ we denote by $w(i)$ the $i$th letter of $w$, by $w\pref i$ the prefix of $w$ up to position $i$ included, and by $w\suf i$ the suffix of $w$ from position $i$ included.
Given two words $u,v$ we write $u\preceq v$ (resp. $u\prec v$) to denote that $u$ is a (resp. strict) prefix of $v$, and we write $u^{-1}v$ the unique word $w$ such that $uw=v$. A \emph{language} over an alphabet $A$ is a subset of $A^*$. A \emph{word relation} $R$ (or \emph{transduction}) over alphabets $A,B$ is a subset of $A^*\times B^*$ and we often write $uRv$ to denote $(u,v)\in R$. Let $R(u)=\set{v\mid\ uRv}$, and if $R$ is a partial function from $A^*$ to $B^*$, we rather write $R(u)=v$ instead of $R(u)=\set{v}$. The \emph{composition} of two relations $R$ and $S$ is $R\circ S=\set {(u,w)|\ \exists v,\ uSv\text{ and } vRw}$. The \emph{inverse} of a relation $R$ is $R^{-1}=\set{(v,u)|\ uRv }$. The \emph{identity relation} over an alphabet $A$ is $\id=\set{(u,u)|\ u\in A^*}$. The \emph{domain} and \emph{range} of a relation $R$ are respectively: $\dom(R)=\set{u|\ \exists v,\ uRv}$ and $\ran(R)=\set{v|\ \exists u,\ uRv}$.

We say that a relation $S$ is \emph{finer} than $R$ (or that $R$ is \emph{coarser} than $S$) if for any words $u,v$, $uSv \Rightarrow uRv$, which we denote by $S\subseteq R$.

An equivalence relation $R$ over alphabet $A$ is a relation over alphabets $A,A$ such that it is reflexive ($\id\subseteq R$), symmetric ($R^{-1}\subseteq R$) and transitive ($R\circ R \subseteq R$). Taking the terminology of \cite[Sec.~2]{Johnson86}, the \emph{(equivalence) kernel} of a total function $f:A^*\rightarrow B^*$ is the equivalence relation $\ker(f)=\set{(u,v)|\ f(u)=f(v)}=f^{-1}\circ f$. A \emph{canonical function} for an equivalence relation $R$ is a function $f$ such that $\ker (f)=R$.
The \emph{transitive closure} of a relation $R$, denoted by $R^+$, is the finest transitive relation coarser than $R$.
Given two equivalence relations $S \subseteq R$ then any equivalence class of $R$ is a union of equivalence classes of $S$ and the \emph{index} of $S$ with respect to $R$ is the supremum of the number of equivalence classes of $S$ included in a unique equivalence class of $R$. We extend the notion of index to arbitrary relations $S\subseteq R$: the index of $S$ with respect to $R$ is the value $\sup_{\begin{smallmatrix}
    {u,T\subseteq R(u)}\\
    \forall v\neq w\in T,\  v{\cancel S}w
\end{smallmatrix}} |T|$.
We denote by $S\subseteq_k R$ that the index of $S$ with respect to $R$ is at most $k$, by $S\subseteq_\finite R$ that the index of $S$ with respect to $R$ is finite, and by $S\subseteq_\infty R$ that the index of $S$ with respect to $R$ is infinite.

The \emph{valuedness} of a relation $R$  is the supremum of the cardinal of the image set of a word, \ie  $\sup_{u}|R(u)|$.

\subparagraph{Automata and transducers}

A \emph{finite automaton} (or just automaton) over an alphabet $A$ is a tuple $\A=\tuple{Q,\Delta, I, F}$ where $Q$ is a finite set of \emph{states}, $\Delta\subseteq Q\times A\times Q$ is a finite \emph{transition relation} and $I,F\subseteq Q$ are the sets of \emph{initial states} and \emph{final states}, respectively.
A \emph{run} of $\A$ over a word $w\in A^*$ is a word $r\in Q^*$ of length $|w|+1$ such that for $1\leq i\leq |w|$, $\tuple{r(i),w(i),r(i+1)}\in \Delta$.
We use the notation $p\xrightarrow{w}_\A q$ (or just $p\xrightarrow{w} q$ when $\A$ is clear from context) to denote that there exists a run $r$ of $\A$ over $w$ such that $r(1)=p$ and $r(|r|)=q$.
Let $r$ be a run of $\A$, if $r(1)\in I$ then $r$ is called \emph{initial}, if $r(|r|)\in F$ then $r$ is called \emph{final} and a run which is both initial and final is called \emph{accepting}. A word $w$ is \emph{accepted} by $\A$ if there is an accepting run over it and the set of words accepted by $\A$ is called the \emph{language recognized} by $\A$ and denoted by $\sem{\A}$. A language is called \emph{rational} if it is recognized by some automaton.

An automaton is called \emph{deterministic} if it has a unique initial state, and for any pair of transitions $(p,a,q_1),(p,a,q_2)\in \Delta$ we have $q_1=q_2$.

A \emph{finite transducer} over alphabets $A,B$ is an automaton over $A^*\times B^*$. We define the natural projections $\pi_A:(A^*\times B^*)^*\rightarrow A^*$ and $\pi_B:(A^*\times B^*)^*\rightarrow B^*$. We say that a pair of words $(u,v)\in A^*\times B^*$ is \emph{realized} by a transducer $\T$ if there exists a word $w$ such that $\T$ has an accepting run $r$ over $w$, $\pi_A(w)=u$ and $\pi_B(w)=v$, and we write $(u,v)\in \sem \T$ with $\sem \T$ denoting the \emph{relation realized} by $\T$. A relation realized by a transducer is called \emph{rational}.
Given a transducer $\T=\tuple{Q,\Delta,I,F}$ we define $\pi_A(\T)$ the \emph{input automaton} of $\T$ by $\tuple{Q,\pi_A(\Delta),I,F}$, where $\pi_A(\Delta)=\set{(p,a,q)|\ \exists b\in B^*\ (p,a,b,q)\in \Delta}$. A transducer is called \emph{real-time} if its transitions are over the alphabet $A\times B^*$ and \emph{letter-to-letter} if its transitions are over $A\times B$. A real-time transducer whose input automaton is deterministic is called \emph{sequential} and the function it realizes is also called sequential.
We say that a relation $R$ is \emph{length-preserving} if for any words $u,v$, $uRv \Rightarrow |u|=|v|$. A letter-to-letter transducer realizes a length-preserving relation and it is known that any length-preserving rational relation can be given as a letter-to-letter transducer. However, one can easily see that a sequential length-preserving function cannot in general be given as a letter-to-letter sequential transducer. For instance the function mapping $aa$  to $aa$ and $ab$ to $bb$ is sequential and length-preserving yet cannot be given as a sequential letter-to-letter transducer.

\subparagraph{Classes of rational equivalence relations}

We define classes of equivalence relations: \rateq the class of all rational equivalence relations, \kerrat the class of relations which are kernels of rational functions and \kerseq the class of relations which are kernels of sequential functions. For each of the previous classes $\mathbf C$, we define $\mathbf C^\lp$ as the class of \emph{length-preserving} relations of $\mathbf C$. Similarly we define $\mathbf C^\ltl$ by restricting to letter-to letter transducers, and we have obviously that $\mathbf C^\ltl \subseteq \mathbf C^\lp$. For instance $\rateq^\ltl$ is the class of equivalence relations which are given by letter-to-letter transducers while $\kerseq^\ltl$ is the class of relations which are kernels of letter-to-letter sequential transducers.
Fig.~\ref{fig:classes} gives the relative inclusions of the classes considered in this article, and a similar one can be found in \cite[Fig.~1]{Johnson86}.
\begin{figure}

    \centering
   \begin{tikzpicture}[scale=.9]

    \node at (0,1) {\small{General case}};
    \node at (0,0) {\small{\rateq}};
    \node[rotate=90] at (0,-.5) {\small{$\subseteq$}};
    \node at (0,-1) {\small{\kerrat}};
    \node[rotate=90] at (0,-1.5) {\small{$\subsetneq$}};
    \node at (0,-2) {\small{\kerseq}};

    \draw[dotted] (1.5,1) -> (1.5,-2.5);
    \draw[very thin] (-1,.7) -> (6,.7);

    \node at (4,1) {\small{Length-preserving}};
    \node at (4,0) {\small{$\begin{array}{cc}\kerrat^\lp=\kerrat^\ltl\\ =\rateq^\lp=\rateq^\ltl\quad\end{array}$}};
    \node[rotate=90] at (4,-.5) {\small{$\subsetneq$}};
    \node at (4,-1) {\small{$\kerseq^\lp$}};
    \node[rotate=90] at (4,-1.5) {\small{$\subsetneq$}};
    \node at (4,-2) {\small{$\kerseq^\ltl$}};
   \end{tikzpicture}
   \caption{Classes of rational equivalence classes.}
   \label{fig:classes}
\end{figure}

It is not known whether the classes \rateq and \kerrat are equal or not. The generic problem we want to study is: given a rational equivalence relation, can we effectively decide if it is in \kerseq?
Let $R$ be a length-preserving equivalence relation given by a transducer $\T$, we know (\textit{e.g.}~from \cite[Thm.~5.1]{Johnson85}) that there is a canonical function given by a transducer which maps any word to the minimum, for the lexicographic order, of its equivalence class. Hence we have that $\rateq^\ltl=\rateq^\lp=\kerrat^\ltl=\kerrat^\lp$ and $\kerseq^\ltl \varsubsetneq \kerseq^\lp$, as we have seen above.
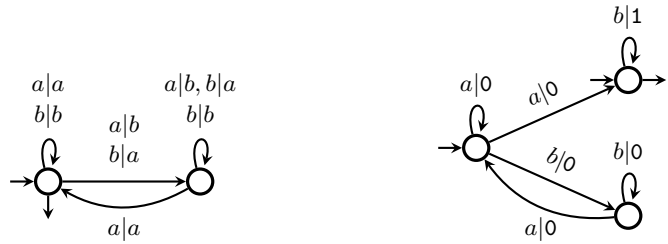
\begin{figure}
    \centering
    \begin{tikzpicture}
        \node[circle,draw,very thick] (0) at (0,0) {};
        \node[circle,draw,very thick] (1) at (2,0) {};

    \draw[>=stealth,thick,->] (0) edge node[above] {\small{$\begin{array}{c} a|b \\  b|a \end{array}$}} (1);
    \draw[>=stealth,thick,->] (0) edge[loop above] node[above] {\small{$\begin{array}{c} a|a \\  b|b\end{array}$}} (0);
    \draw[>=stealth,thick,->] (1) edge[loop above] node[above] {\small{$\begin{array}{c} a|b,  b|a\\b|b \end{array}$}} (1);
    \draw[>=stealth,thick,->] (1) edge[bend left] node[below] {\small{$a|a $}} (0);
    \draw[>=stealth,thick,<-] (0) edge  +(-.5,0);
    \draw[>=stealth,thick,->] (0) edge  +(0,-.5);

    \end{tikzpicture}
    \hspace{2cm}~
    \begin{tikzpicture}
        \node[circle,draw,very thick] (0) at (0,0) {};
        \node[circle,draw,very thick] (1) at (2,-.9) {};
        \node[circle,draw,very thick] (2) at (2,.9) {};
        \draw[>=stealth,thick,->] (0) edge node[above,sloped] {\small{$b|\mathtt 0$}} (1);
        \draw[>=stealth,thick,->] (1) edge[bend left] node[below] {\small{$a| \mathtt 0$}} (0);
        \draw[>=stealth,thick,->] (0) edge[] node[above,sloped] {\small{$a|\mathtt 0$}} (2);

        \draw[>=stealth,thick,->] (0) edge[loop above] node[above] {\small{$a|\mathtt 0$}} (0);
        \draw[>=stealth,thick,->] (1) edge[loop above] node[above] {\small{$b|\mathtt 0$}} (1);
        \draw[>=stealth,thick,->] (2) edge[loop above] node[above] {\small{$b|\mathtt 1$}} (2);

        \draw[>=stealth,thick,<-] (0) edge  +(-.5,0);
        \draw[>=stealth,thick,<-] (2) edge  +(-.5,0);
        \draw[>=stealth,thick,->] (2) edge  +(.5,0);

    \end{tikzpicture}
    \caption{On the left a transducer recognizing an equivalence relation. On the right a transducer realizing a canonical function for it. Two words are equivalent if their last $a$ is at the same position.}
    \label{fig:lasta}
\end{figure}
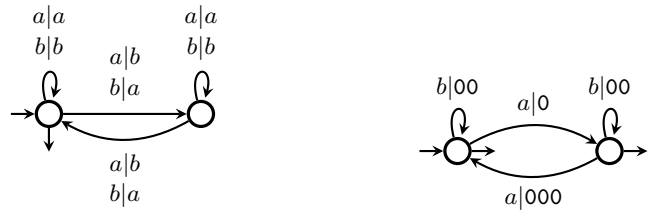
\begin{figure}
    \centering
    \begin{tikzpicture}
        \node[circle,draw,very thick] (0) at (0,0) {};
        \node[circle,draw,very thick] (1) at (2,0) {};

    \draw[>=stealth,thick,->] (0) edge[] node[above] {\small{$\begin{array}{c} a|b \\  b|a\end{array}$}} (1);
    \draw[>=stealth,thick,->] (0) edge[loop above] node[above] {\small{$\begin{array}{c} a|a \\  b|b\end{array}$}} (0);
    \draw[>=stealth,thick,->] (1) edge[loop above] node[above] {\small{$\begin{array}{c} a|a \\  b|b\end{array}$}} (1);
    \draw[>=stealth,thick,->] (1) edge[bend left] node[below] {\small{$\begin{array}{c} a|b \\  b|a\end{array}$}} (0);
    \draw[>=stealth,thick,<-] (0) edge  +(-.5,0);
    \draw[>=stealth,thick,->] (0) edge  +(0,-.5);

    \end{tikzpicture}
    \hspace{2cm}~
    \begin{tikzpicture}
        \node[circle,draw,very thick] (0) at (0,0) {};
        \node[circle,draw,very thick] (1) at (2,0) {};
        \draw[>=stealth,thick,->] (0) edge[bend left] node[above] {\small{$a|\mathtt 0$}} (1);
        \draw[>=stealth,thick,->] (1) edge[bend left] node[below] {\small{$a|\mathtt{000}$}} (0);

        \draw[>=stealth,thick,->] (0) edge[loop above] node[above] {\small{$b|\mathtt{00}$}} (0);
        \draw[>=stealth,thick,->] (1) edge[loop above] node[above] {\small{$b|\mathtt{00}$}} (1);

        \draw[>=stealth,thick,<-] (0) edge  +(-.5,0);
        \draw[>=stealth,thick,->] (0) edge  +(.5,0);
        \draw[>=stealth,thick,->] (1) edge  +(.5,0);

    \end{tikzpicture}
    \caption{An equivalence relation and a sequential canonical function for it. Two words are equivalent if their number of $a$'s is the same modulo $2$.}
    \label{fig:evena}
\end{figure}
We give in Fig.~\ref{fig:lasta} an example of length-preserving rational equivalence relation $R$, and we exhibit a rational canonical function for it. This equivalence relation is not in $\kerseq$ and this can be shown using the characterization we prove in Sec.~\ref{sec:seq}. Intuitively, one has to \emph{guess} when reading an $a$ if it is the last one or not, which cannot be done sequentially.
In Fig.~\ref{fig:evena}, we exhibit an equivalence relation which is length-preserving and is the kernel of a sequential function. However it is not the kernel of a \emph{letter-to-letter} sequential function, which we will be able to show using the characterization from Sec.~\ref{sec:ltl}.

\section{Kernels of sequential letter-to-letter functions}
\label{sec:ltl}

The goal of this section is to characterize relations which are kernels of sequential letter-to-letter functions.
First, in Sections~\ref{subsec:synt} and~\ref{subsec:pref-close} we give two necessary conditions for a relation to be in $\kerseq^\ltl$.
Then in Sec.~\ref{subsec:cons-seq-ltl} we provide an algorithm to construct a sequential letter-to-letter canonical function when the two aforementioned conditions are satisfied, showing that they are indeed sufficient and thus characterize $\kerseq^\ltl$.
Finally in Sec.~\ref{subsec:dec-seq-ltl}, we state the characterization established before and show that it is decidable.

\subsection{Syntactic congruence}
\label{subsec:synt}
We start by introducing a notion of syntactic congruence associated with an equivalence relation, which will prove crucial throughout the paper.
Given a relation $R$, we define $S_R$ the \emph{syntactic congruence} of $R$ by $uS_Rv$ if for any word $w$, we have $uwRvw$. In particular $S_R$ is finer than $R$ and $S_R$ is a (right) congruence meaning that if $uSv$ then for any letter $a$ we have $uaS_Rva$. Furthermore, if $R$ is an equivalence relation then so is $S_R$.

We now exhibit a first necessary condition to be in $\kerseq$, and \textit{a fortiori} in $\kerseq^\ltl$.
\begin{proposition}
\label{prop:finite-index}
Let $R$ be an equivalence relation. If $R\in \kerseq$ then $S_R$ has finite index with respect to $R$.
\end{proposition}
\begin{proof}
Let $\T$ be a sequential transducer realizing a function $f$ whose kernel is $R$, and let $n$ be the number of states of $\T$.
Let $uRv$, then we have $f(u)=f(v)$.
Furthermore, if $u,v$ reach the same state in $\T$, since $\T$ is sequential, $f(uw)=f(vw)$ for any word $w$ which means that $uS_Rv$.
Let $u_1Ru_2R\ldots R u_{n+1}$. By a pigeon-hole argument, there must be two indices $1\leq i<j\leq n+1$, such that $u_{i}$ and $u_{j}$ reach the same state in $\T$, hence $u_{i}S_Ru_{j}$. Thus we have shown that the index of $S_R$ with respect to $R$ is less than $ n$, and is thus finite.
\end{proof}
\begin{figure}
    \centering
    \begin{tikzpicture}
        \node[circle,draw,very thick] (0) at (0,0) {};
        \node[circle,draw,very thick] (1) at (2,-.9) {};
        \node[circle,draw,very thick] (2) at (2,.9) {};

    \draw[>=stealth,thick,->] (0) edge node[below] {\small{$\begin{array}{c} a|b \\  b|a \end{array}$}} (1);
    \draw[>=stealth,thick,->] (0) edge[loop above] node[above] {\small{$\begin{array}{c} a|a \\  b|b\end{array}$}} (0);
    \draw[>=stealth,thick,->] (1) edge[loop above] node[above] {\small{$a,b|a,b$}} (1);
    \draw[>=stealth,thick,->] (0) edge node[above,sloped] {\small{$c|c$}} (2);
    \draw[>=stealth,thick,->] (2) edge[loop above] node[right] {\small{$\begin{array}{c} a|a ,  b|b\\ c|c\end{array}$}} (2);
    \draw[>=stealth,thick,<-] (0) edge  +(-.5,0);
    \draw[>=stealth,thick,->] (0) edge  +(0,-.5);
    \draw[>=stealth,thick,->] (1) edge  +(.5,0);
    \draw[>=stealth,thick,->] (2) edge  +(.5,0);

    \end{tikzpicture}
    \caption{An equivalence relation not in $\kerseq$.}
    \label{fig:indexinf}
\end{figure}
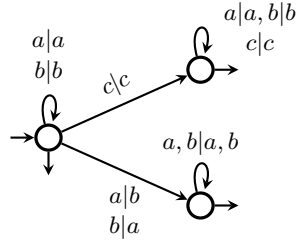
We give in Fig.~\ref{fig:indexinf} an example of a length-preserving equivalence relation such that its syntactic congruence does not have a finite index with respect to it.
Two different words are never syntactically equivalent, however two words of same length without any $c$s are equivalent.
Thus by Prop.~\ref{prop:finite-index}, this relation is not in \kerseq.

In the next two propositions, we show that 1) the syntactic congruence can be computed for a relation in $\rateq^\lp$ and 2) that the finiteness of its index can also be decided.
\begin{proposition}
    \label{prop:comp-synt}
Let $R$ be an equivalence relation given as a pair-deterministic letter-to-letter transducer. One can compute a transducer recognizing its syntactic congruence in \ptime.
\end{proposition}

\begin{proof}
Let $R$ be given by a letter-to-letter pair-deterministic transducer $\R$, and let $S_R$ denote its syntactic congruence.
Let $(u,v)$ be a pair of words of equal length, and let us denote by $p$ the state reached in $\R$ after reading $(u,v)$. Then $uS_Rv$ if and only if the automaton $\R_p$ (obtained by taking $p$ as initial state) recognizes a reflexive relation. This property can be easily checked and thus $S_R$ is obtained by taking $\R$ and restricting the final states to states $p$ such that $\R_p$ recognizes a reflexive relation.
\end{proof}

    \begin{proposition}
        \label{prop:dec-fin-index}

    Let $R$ be a rational relation given as a transducer $\R$, and let $f$ be a rational function given by a transducer $\F$ such that $S=\ker(f)$ is finer than $R$. Then one can decide if $S$ has finite index with respect to $R$ in \ptime.
    \end{proposition}
    
    \begin{proof}
    Let $f$ be a rational function such that $\ker(f)=S$. We show that the index of $S$ with respect to $R$ is equal to the valuedness of $T=f \circ R$.
    We want to show that for any $u$, $|T(u)|=\max_{\begin{smallmatrix} 
        {u,X\subseteq R(u)}\\
        \forall v\neq w\in X,\  v{\cancel S}w
    \end{smallmatrix}} |X|$.
    
    Let $X\subseteq R(u)$ be such that $\forall v\neq w\in X,\  v{{\cancel S}}w$. Then $f$ is injective over $X$ since $f$ maps words to the same value if and only if they are $S$ equivalent, thus $|X|=|f(X)|$. Moreover $f(X)\subseteq T(u)$, which means that $|T(u)|\geq\max_{\begin{smallmatrix} 
        {u,X\subseteq R(u)}\\
        \forall v\neq w\in X,\  v{\cancel S}w
    \end{smallmatrix}} |X|$.
    
    For each $v\in T(u)$, we can find a word $v'\in f^{-1}(v)$ (for instance the minimum word in the lexicographic order). Let $X$ be the set of these words, we have by construction $|X|=|T(u)|$. Moreover, for each pair of distinct words $v',w'\in X$, we have $f(v')\neq f(w')$ and thus in particular $v' {{\cancel S}} w'$. Thus we have shown $|T(u)|\leq\max_{\begin{smallmatrix} 
        {u,X\subseteq R(u)}\\
        \forall v\neq w\in X,\  v{\cancel S}w
    \end{smallmatrix}} |X|$.
    
    Hence the index of $S$ with respect to $R$ is equal to the valuedness of $T$.
    Since finite valuedness can be decided in \ptime \cite[Thm.~3.1]{Weber89}, then one can decide if $S$ has finite index with respect to $R$, also in \ptime.
    \end{proof}

    \begin{corollary}\label{cor:dec-finite-index}
        Let $R\in \rateq^\lp$, one can decide if its syntactic congruence $S_R$ has finite index with respect to it.
    \end{corollary}
    \begin{proof}
        From Prop.~\ref{prop:comp-synt} we can compute a transducer realizing $S_R$. According to \cite[Thm.~5.1]{Johnson85}, we can even compute a transducer realizing a function whose kernel is $S_R$. Hence from Prop.~\ref{prop:dec-fin-index} we can decide the finiteness of the index of $S_R$ with respect to $R$. 
    \end{proof}

\subsection {Prefix closure}
\label{subsec:pref-close}
Here we consider a second necessary condition of relations in $\kerseq^\ltl$, namely that they are prefix-closed.

The \emph{prefix closure} of a relation $R$ is the relation $P_R$ defined by $uP_Rv$ if there exists $u',v'$, with $|u'|=|v'|$, such that $uu'Rvv'$.
A relation is called \emph{prefix-closed} if it is equal to its prefix closure. We often say that $u,v$ are \emph{equivalent in the future} when $uP_R v$.
\begin{proposition}
    \label{prop:prefix-closed}
Let $R$ be an equivalence relation. If $R\in \kerseq^\ltl$ then $R$ is prefix-closed.
\end{proposition}
\begin{proof}
Let $R$ be an equivalence relation, let $f$ be realized by a transducer in $\kerseq^\ltl$ such that $\ker f=R$ and let $P_R$ denote the prefix closure of $R$.
Let $uP_Rv$, then there exist $u',v'$  with $|u'|=|v'|$ such that $f(uu')=f(vv')$.
Since $f$ is letter-to-letter sequential, we have $f(u)\preceq f(uu')$, $f(v)\preceq f(vv')$ and $|f(u)|=|f(v)|$ which means that $f(u)=f(v)$. Hence $uRv$, $R=P_R$ and $R$ is prefix-closed.
\end{proof}
The equivalence relations given in Figures~\ref{fig:lasta} and~\ref{fig:evena} are \emph{not} prefix-closed, which explains why they are not in $\kerseq^\ltl$, according to Prop.~\ref{prop:prefix-closed}.

\subsection{Construction of a canonical function}
\label{subsec:cons-seq-ltl}

The main technical lemma of this section says that the two necessary conditions given above are sufficient:
\begin{lemma}
    \label{lem:cons-seq}
Let $R\in \rateq^\lp$ be prefix-closed with a finite index syntactic congruence with respect to it. Then we can construct a sequential letter-to-letter transducer whose kernel is $R$.
\end{lemma}  

\begin{proof}
    Let $R\in \rateq^\lp$ be given by a  transducer $\R=\tuple{Q,\Delta_R, I, F_R}$ which is letter to letter, over the alphabet $A\times A$, and such that $S_R\subseteq_k R = P_R$ for some $k\in\nat$. Without loss of generality, we assume that $\R$ is deterministic.
    A state of $\R$ will be called \emph{diagonal} if the identity is accepted from that state, and let $D\subseteq F_R$ be the set of diagonal states.
    According to Prop.~\ref{prop:comp-synt}, we can obtain a letter-to-letter transducer $\Sy$ realizing $S$ (just by setting $D$ as the set of final states).

    Our goal is to define a sequential letter-to-letter transducer $\T$ whose kernel is the relation $R$. The main idea to obtain this construction is to distinguish three kinds of relationships between two words: 1) $uSv$ 2) $u{{\cancel S}}v$ and $uRv$ and 3) $u{\cancel R} v$. Then the key idea, as seen in the proof of Prop.~\ref{prop:finite-index}, is that two words in case number 2) \emph{cannot} end up in the same state in $\T$. Two words in situation number 1) might as well reach the same state in $\T$ since they have the exact same behavior. Then two words in situation 3) may or may not reach the same state, it does not matter since their image by $\T$ should be different.

    For each equivalence class of $R$ containing $l\leq k$ different $S$-equivalence classes we define $l$ distinct states. The states will be pairs $\tuple{M,i}$ where $M\in \M_l(Q)$ is an $l\times l$ square matrix with values in $Q$, the state space of $\R$, and $i\in \set{1,\ldots,l}$. Let $u_1,\ldots, u_l$ be the least lexicographic representatives of the $l$ $S$-equivalence classes, in lexicographic order.
    Then $M(i,j)=p$ if $p$ is the state reached in $\R$ after reading $(u_i,u_j)$. Then the state $\tuple{M,i}$ is supposed to be the state reached after reading $u_i$, or any other $S$-equivalent word. Let us remark that the reachable states will only contain matrices where all states are accepting, \ie with values in $F_R$. Moreover, all values on the diagonal are in $D$.

    Let us define a sequential transducer $\T=\tuple{Q_\M,\Delta,\set{(M_0,1)}}$ whose kernel will be the relation $R$ (we don't specify the final states since all states are final). As we have seen, we define $Q_\M=\bigcup_{l\in\set{1,\ldots,k}}\M_l(Q)\times \set{1,\ldots,l}$. Since the word $\epsilon$ is the only word of length $0$, it is alone in its $R$ and $S$-equivalence classes, hence $M_0$ is the $1\times 1$ matrix with value $q_0$ the initial state of $\R$.
    We have left to define $\Delta$ and then show that the construction is correct. This will be done by induction on the length of the words. More precisely, let us state the induction hypothesis for words of length $n$:
    \begin{itemize}
        \item[$Hn.1$:] Let $u_1,\ldots,u_l$ be the minimal representatives of the $S$-equivalence classes of some $R$-equivalence class, of words of length $\leq n$.
        Then any word $uSu_i$ with $i\in \set{1,\ldots,l}$, reaches the state $(M,i)$ where $M(j,j')$ is the state reached in $\R$ by reading the pair $(u_j,u_{j'})$.
        \item[$Hn.2$:] Two words, of length $\leq n$, are $R$-equivalent if and only if their outputs in $\T$ are equal.
    \end{itemize}

    This trivially holds for the word of length $0$, and let us assume that it holds for words of length $\leq n$.
    Let $u_1,\ldots,u_l$ be the minimal representatives of the $S$-equivalence classes of some $R$-equivalence class, of words of length $n$.
    Let us consider the corresponding matrix $M\in \M_l(Q)$.
    
    Let us define an equivalence relation $\sim_R$ over $\set{1,\ldots,l}\times A$ which will separate word which are no longer $R$-equivalent. Let $q_{i,j,a,b}$ be the state reached in $\R$ from $M(i,j)$ by reading $(a,b)$. 
    Two pairs $(i,a),(j,b)$ are $\sim_R$-equivalent if $q_{i,j,a,b} \in F_R$.
    By $Hn.1$ we know that this is indeed an equivalence relation.
    We define a second equivalence relation $\sim_S$. Two pairs $(i,a),(j,b)$ are equivalent if $q_{i,j,a,b} \in D$.
    Finally, we consider a linear order on $\set{1,\ldots,l}\times A$ which is just the lexicographic order (with some fixed order over $A$).

    Let $(i,a)\in \set{1,\ldots,l}\times A$, let us consider the set of minimal $\sim_S$-representatives of the $\sim_R$-equivalence class of $(i,a)$: 
    $$I_R=\set{(j,b)|\ (j,b)\sim_R (i,a)\text{ and }\forall (j',b')<(j,b),\ (j,b)\mathrel{{\cancel \sim}_S}(j',b')}$$
    Let $l'$ denote the cardinal of $I_R$, \ie the number of $\sim_S$ equivalence classes in the $\sim_R$-equivalence class of $i$. We define the state $(N,j)$ and the output $b\in B$ such that $((M,i),(a,b),(N,j))\in \Delta$.
    The output $b$ is defined by $\min I_R$.
    The matrix $N$ has dimension $l'$ and let $(i_1,a_1),\ldots, (i_{l'},a_{l'})$ be the elements of $I_R$ in increasing order. The matrix $N$ is defined by $N(j,j')=p$ where $p$ is the state reached from $M(i_j,i_{j'})$ by reading $(a_j,a_{j'})$. Let $j$ be the index such that $(i,a)\sim_S (i_j,a_j)$, then we have $((M,i),(a,b),(N,j))\in \Delta$.
    
    Let us show $Hn+1.1$. Let $uSu_ja$, we need to show that $u$ reaches the state $(N,j)$. Let $vc=u$, with $c\in A$. Since $vcSu_ja$, we have $vcSu_ja$, which means that $vRu_j$, since $R$ is prefix closed. hence there exists $u_{j'}$ such that $vS u_{j'}$.
    This means that we have $u_{j'}cSu_ja$ and $u_{j'}c\geq u_ja$ in the lexicographic order. By induction hypothesis, $v$ reaches the state $(M,j')$, and by construction we have $((M,j'),(c,b),(N,j))\in \Delta$.

    We now show $Hn+1.2$.
    Let $v_1=w_1a_1,v_2=w_2a_2$ be two words of length $n+1$, with $a_1,a_2\in A$. If $v_1Rv_2$, then $w_1Rw_2$ since $R$ is prefix closed. By induction hypothesis, the outputs over $w_1$ and $w_2$ are the same. Moreover, by construction of $\Delta$, the final outputs reading $a_1$ and $a_2$, respectively, are the same. If $w_1{{\cancel R}}w_2$, then by induction, their outputs are different, and so are the outputs over $v_1,v_2$. The only remaining case is when $w_1Rw_2$ and $v_1{{\cancel R}}v_2$. By induction, we have that the outputs over $w_1,w_2$ are the same, hence we need to show that the outputs from the letters $a_1$, $a_2$ are different. By the construction of $\Delta$, the outputs are linked with $\sim_R$ equivalence classes, which means that the outputs corresponding to $w_1,a_1$ and $w_2,a_2$ are different.
\end{proof}

\subsection{Characterization of $\kerseq^\ltl$ and decidability}
\label{subsec:dec-seq-ltl}

As a corollary we obtain a characterization of $\kerseq^\ltl$.  
\begin{theorem}[Characterization of $\kerseq^\ltl$]
    \label{thm:char-seq}
        Let $R\in \rateq$. The following are equivalent:
        \begin{enumerate}
            \item $R\in \kerseq^\ltl$
            \item $R$ is length-preserving and $S_R\subseteq_\finite R = P_R$
        \end{enumerate}
\end{theorem}

\begin{proof}
1.$\Rightarrow$2.~comes from the results of Prop.~\ref{prop:finite-index} and Prop.~\ref{prop:prefix-closed}.
To obtain 2.$\Rightarrow$1.~we use the construction of Lem.~\ref{lem:cons-seq}.
\end{proof}
From the previous result we get an algorithm deciding if an equivalence relation is in $\kerseq^\ltl$.
\begin{theorem}
    \label{thm:dec-ltl}
        The following problem is decidable.
        \begin{enumerate}
                \item \textbf{Input:} $\R$ a transducer realizing an equivalence relation $R$.
                \item \textbf{Question:} Does $R$ belong to $\kerseq^\ltl$?
        \end{enumerate}
\end{theorem}

\begin{proof}
Without loss of generality, we can assume that $\R$ is a letter-to-letter pair-deterministic transducer.
From Cor.~\ref{cor:dec-finite-index} we can decide if $S_R$ has finite index with respect to $R$.
Deciding if $R$ is prefix-closed, is easy: just check if a reachable state is not final.

According to Thm.~\ref{thm:char-seq}, we thus have an algorithm to decide the problem.
\end{proof}

\section{Kernels of sequential functions}
\label{sec:seq}

We turn to the problem of deciding membership in $\kerseq^\lp$. 
To tackle this we introduce another kind of transducers called \emph{subsequential}, which are transducers allowed to produce a final output at the end of a computation. A \emph{subsequential transducer} over alphabets $A,B$  is a pair $(\T, \fin)$, where $ \fin:F\rightarrow B$ is called the \emph{final output function} ($F$ being the set of final states of $\T$). We denote by \kersseq the class of equivalence relations which are kernels of subsequential functions.

Our results are obtained in two steps. First we exhibit sufficient conditions for being in $\kersseq ^\ltl$ very similar to the characterization of $\kerseq ^\ltl$.
Second we show that $\kersseq ^\ltl=\kerseq ^\lp=\kersseq ^\lp$.

\subsection{Construction for $\kersseq^\ltl$}
When studying relations in $\kersseq^\ltl$, we lose the property of being prefix-closed.
We have to consider instead the transitive closure of the prefix closure.

\begin{theorem}
    \label{thm:cons-kersseq}
Let $R\in \rateq^\lp$, let $P_R$ be the prefix closure of $R$ such that $S_R$ has finite index with respect to $P_R^+$. Then we can construct a subsequential letter-to-letter transducer whose kernel is $R$.
\end{theorem}

\begin{proof}
    
    From Prop.~\ref{prop:comp-synt}, we can obtain a transducer $\Sy$ realizing $S_R$. Let $\Pp$ be a transducer realizing $P_R^+$. Without loss of generality, we assume that $\R,\Sy,\Pp$ are letter-to-letter and deterministic.
    Let us assume that $S_R\subseteq_k P_R^+$.
    We use the algorithm defined in the proof of Lem.~\ref{lem:cons-seq} to obtain a transducer which realizes $P_R^+$, with state space $\bigcup_{l\leq k}\M_l(Q)$, where $Q=Q_\R\times Q_\Pp$, the product of the state spaces of $\R$ and $\Pp$. Using the same construction we can obtain a sequential transducer realising $P_R^+$ with the following properties:

    \begin{itemize}
        \item[$H.1$:] Let $u_1,\ldots,u_l$ be the minimal representatives of the $S_R$-equivalence classes of some $P_R^+$-equivalence class.
        Then any word $uS_Ru_i$ with $i\in \set{1,\ldots,l}$, reaches the state $(M,i)$ where $M(j,j')$ is the state reached in $\R\times\Pp$ by reading the pair $(u_j,u_{j'})$.
        \item[$H.2$:] Two words are $P_R^+$ equivalent if and only if their outputs in $\T$ are equal.
    \end{itemize}
We only need to define a final output function $t:Q_\R\times Q_\Pp\rightarrow B$ which will differentiate words that are $P_R^+$ equivalent but not $R$ equivalent. Let $u_1,\ldots,u_l$ be the minimal representatives of the $S_R$-equivalence classes of some $P_R^+$-equivalence class, and let $M\in \M_l(Q)$ be the corresponding matrix such that $M(i,j)$ is the state reached by reading $(u_u,u_{j})$ in $\R\times\Pp$. Then let us consider the equivalence  relation $\sim_R$ over $\set{1,\ldots,l}$ defined by $i\sim_Rj$ if and only if $u_iRu_j$. Then we define $t(M,i)=\min_{j\sim_R i}j$.

 According to $H.1$ we only need to show that this construction is correct for minimal lexicographic representatives of $S$ classes. Let $u,v$ be two words of same length, and let us assume that $u {\cancel {P_R^+}}v$. Then the images  of $u$ and $v$ are already different, even without taking the final output into account. Let us assume that $u  P_R^+ v$, then $u,v$ have the same image by $\T$.
  If $u {\cancel R}v$ considering that $u, v$ are representatives of their respective $S$ class, we have by definition that $t(M,i_u)\neq t(M,i_v)$, where $(M,i_u)$ and $(M,i_v)$ are the states reached by reading $u$ and $v$, respectively.
  Similarly, we show that if $uRv$, then final outputs are the same which means that the image of $u,v$ by $(\T,t)$ is the same.

\end{proof}

\subsection{Equality of classes}

Let us start by stating the obvious inclusions which are just obtained by syntactic restrictions: $\kersseq^\ltl \subseteq \kersseq^\lp$ and $\kerseq^\lp \subseteq \kersseq^\lp$. 

We now show that one can remove the final outputs by adding modulo counting.
 \begin{lemma}
    \label{lem:sseq-seq}
$\kersseq^\ltl\subseteq \kerseq$
 \end{lemma}
 
 \begin{proof}
 Let $(\T, \fin)$ with $\T=\tuple{Q,\Delta, \set{q_0}, F}$ be a subsequential letter-to-letter transducer over $A,B$ realizing a function $f$, and let $g$ be the function realized by $\T$.
 Let $\sim_\fin$ be an equivalence relation defined over $F$ by $p\sim_\fin q$ if $\fin(p)=\fin(q)$. Let $u,v$ be two words that reach states $p,q$ respectively from $q_0$. Then, $f(u)=f(v)$ if and only if $g(u)=g(v)$ and $p\sim_\fin q$. We know that the number of equivalence classes of $\sim_\fin$ is less than $n=|B|$, so we number the equivalence classes from $1$ to $n$. The main idea is to consider $g^n$ which multiplies in $g$ every occurrence of each letter by $n$, except for the last letter. Then, the number of occurrences of the last letter encodes, modulo $n$, the equivalence class of the state. Hence for any words $u,v$ we have $g^n(u)=g^n(v)$ if and only if  $g(u)=g(v)$ and $p\sim_\fin q$ if and only if $f(u)=f(v)$, which means that the equivalence kernel of $f$ is equal to that of $g^ n$.  

Let us now show that $g^n$ is sequential.
We extend the equivalence relation $\sim_\fin$ arbitrarily to non final states, and to simplify things, we assume that the equivalence class of the initial state is $n$.
Let us define a transducer $\T^n=\tuple{ Q\times B,\Delta^n, \set{(q_0,b_0)}, F\times B}$ realizing $g^n$ (where $b_0$ is some fixed letter in $B$).
Let $p,q\in Q$ with respective equivalence classes $i,j\in \set{1,\ldots, n}$ such that $(p,(a,b),q)\in \Delta$.
For any $c\in B$ we have $((p,c),(a,c^{n-i}b^{j}),(q,b))\in \Delta^n$.
 \end{proof}

We only have left to show that relations in $\kerseq^\lp$ satisfy the sufficient conditions to be in $\kersseq^\ltl$. 

We need a few technical results before showing the main lemma.
The next claim is quite simple and just says that if two words can be equivalent in the future, then they can be equivalent in a near future.
\begin{claim}
    \label{lem:fut-close}
Let $R\in \rateq^\ltl$ and let $P_R$ be the prefix closure of $R$. There exists $D\geq 0$ such that for all $u,v$ with $uP_Rv$ there exists $w_1,w_2$ with $|w_1|=|w_2|\leq D$ and $uw_1Ruw_2$.
\end{claim}
\begin{proof}
Let $\R$ be a letter-to-letter transducer recognizing $R$.
We assume without loss of generality that $\R$ is pair-deterministic. Then the relation $P_R$ is recognized by $\R'$ which is just $\R$ where all states that can reach a final state become final.
Let $D$ be the number of states of $\R$. If $uP_Rv$ there exists $w_1,w_2$ with $|w_1|=|w_2|\leq D$ and $uw_1Ruw_2$.
\end{proof}
This next statement is a quite simple consequence of the previous one. If two words are equivalent in the future, then their images by a subsequential kernel function have to be close too.
\begin{claim}
    \label{lem:fut-out-close}
    $R\in \kersseq^\lp$, let $f$ be a subsequential function such that $\ker(f)=R$ and let $P_R$ be the prefix closure of $R$. There exists $\delta\geq 0$ such that for all $u,v$ with $uP_Rv$,  $||f(u)|-|f(v)||\leq \delta$.
\end{claim}

\begin{proof}
    $R\in \kersseq^\lp$, let $f$ be a subsequential function such that $\ker(f)=R$ and let $P_R$ be the prefix closure of $R$.
    Let $(\T,\final)$ be a subsequential transducer realizing $f$ and let $K$ be the maximal size of an output of $(\T,\final)$. According to Lem.~\ref{lem:fut-close}, we know that there exists $D$ such that, if $uP_Rv$, then there exists $w_1,w_2$ with $|w_1|=|w_2|\leq D$ and $uw_1Ruw_2$. This means that $f(uw_1)=f(vw_2)$, and thus $||f(u)|-|f(v)||\leq 2KD$.
\end{proof}
The next lemma is the most technical part of this section, and its proof is given in App.~\ref{app:fplus-close} due to a lack of space. It says that if two words are \emph{transitively} future equivalent, then their images by a subsequential canonical function have to be close.
\begin{lemma}
    \label{lem:fplus-close}
    $R\in \kersseq^\lp$, let $f$ be a subsequential function such that $\ker(f)=R$ and let $P_R$ be the prefix closure of $R$. There exists $D\geq 0$ such that for all $u,v$ with $uP_R^+v$,  $||f(u)|-|f(v)||\leq D$.
\end{lemma}
The previous lemma shows that two words that are \emph{transitively} future equivalent must have close output from a subsequential canonical function. By a pigeon-hole argument we obtain in the next corollary that a relation in $\kersseq^\lp$ must have finite index with respect to the transitive closure of the future equivalence.

\begin{corollary} 
    \label{cor:char-kerseq}
    Let $R\in \kersseq^\lp$, and let $P_R$ denote the prefix closure of $R$. Then $S_R$ has finite index with respect to $P_R^+$.
\end{corollary}
\begin{proof}
    Let $(\T,\fin)$ be a subsequential transducer realizing $f$ such that $\ker f=R$, and let $n$ be the number of states of $\T$. According to Lem.~\ref{lem:fplus-close}, there exists $D$ such that for all $u,v$ with $uP_R^+v$,  $||f(u)|-|f(v)||\leq D$. Let $N=|B|^{D+1}$ and let $u_1P_R^+u_2P_R^+\ldots P_R^+ u_{(n+1)N}$.
For all $i,j\in\set{1,\ldots,(n+1)N}$, $||f(u_i)|-|f(u_j)||\leq D$. This means that the set $\set{f(u_i)|\ 1\leq i\leq (n+1)N}$ has cardinality less than $N$. Thus there exists $i_1<\ldots<i_{n+1}$ such that $f(u_{i_1})=\ldots=f(u_{i_{n+1}})$.
One can see that there must be two indices $1\leq j<k\leq n+1$, such that $u_{i_{j}}$ and $u_{i_{k}}$ reach the same state in $\T$, hence $u_{i}Ru_{j}$, and even $u_{i}S_Ru_{j}$. Thus we have shown that the index of $S_R$ with respect to $P_R^+$ is less than $ (n+1)N$, and is thus finite.
\end{proof}

\begin{proposition}[Equality of classes]
    \label{prop:eq-classes}
    The following classes of equivalence relations are identical:
    \begin{enumerate}
        \item $\kerseq^\lp$
        \item $\kersseq^\ltl$
        \item $\kersseq^\lp$
    \end{enumerate}
\end{proposition}

\begin{proof}
    The proof is given in Fig.~\ref{fig:proof}. The arrows represent class inclusion. Black arrows are trivial syntactic restrictions.
    \begin{figure}
  \centering
    \begin{tikzpicture}[scale=.9]
        \node (sl) at (0,0) {\small{$\kersseq^\ltl$}};
        \node (p) at (4,0) {\small{$\kerseq^\lp$}};
        \node (sp) at (2,2) {\small{$\kersseq^\lp$}};

        \draw[very thick,>=stealth,->] (sl) edge[] (sp);
        \draw[very thick,>=stealth,->] (p) edge[] (sp);

    \draw[very thick,>=stealth,->,color=red] (sp) edge[bend right] node[above left] {\small{\begin{tabular}{c}Cor.~\ref{cor:char-kerseq} +\\ Thm.~\ref{thm:cons-kersseq}\end{tabular}}}(sl);
        \draw[very thick,>=stealth,->,color=red] (sl) edge[]  node[above] {\small{Lem.~\ref{lem:sseq-seq}}} (p);
    \end{tikzpicture}
        \caption{Proof of Prop.~\ref{prop:eq-classes}.}
        \label{fig:proof}
    \end{figure}
    
\end{proof}

\section{Deciding membership in \kerseq}

We show here that knowing if a rational equivalence relation is in $\kerseq$ is an undecidable problem, and this even if the relation is length-preserving.
The trouble lies with computing the equivalence relation $P_R^+$. Indeed, transitive closures of even very simple relations are known not to be computable (the next configuration of a Turing machine can be computed by a simple transduction).

Let us first state a characterization of $\kerseq^\lp$, by combining the results of the previous subsections.
\begin{theorem}[Characterization of $\kerseq^\lp$]
    \label{thm:char-seq2}
    Let $R\in \rateq^\ltl$. The following are equivalent:
    \begin{enumerate}
        \item $R\in \kerseq$
        \item $S_R\subseteq_\finite R\subseteq_\finite  P_R^+$
        \item $S_R\subseteq_\finite R\subseteq_\finite  P_R$ and $\exists k\ P_R^k=P_R^{k+1}$
    \end{enumerate}
\end{theorem}
\begin{proof}
    $1\rightarrow 2$.
    Let $R\in \rateq^\ltl$.
    Let us first assume that $R\in \kerseq$. Then according to Cor.~\ref{cor:char-kerseq}, we have $S_R\subseteq_\finite R\subseteq_\finite  P_R^+$.

    $2\rightarrow 1$.
    Conversely, let us assume that $S_R\subseteq_\finite R\subseteq_\finite  P_R^+$. According to Theorem~\ref{thm:cons-kersseq}, we can construct a subsequential letter-to-letter transducer whose kernel is $R$. From Lem.~\ref{lem:sseq-seq}, we have $R\in \kerseq$.

    $2\rightarrow 3$.
    Let us assume $S_R\subseteq_\finite R\subseteq_\finite  P_R^+$.
    In particular $S_R\subseteq_\finite R\subseteq_\finite  P_R$. Let us assume that $S_R\subseteq_N  P_R^+$.
    Let $uP_R^+ v$ and let $u=u_0P_Ru_1\ldots P_R u_m=v$ be a chain of minimal length $m$. If we assume $m>N$, then there must exist $i,j\leq m$ such that $u_iS_Ru_j$. Since $u_i$ and $u_j$ are \emph{syntactically} equivalent, this means that $u_iP_Rw \Leftrightarrow u_jP_Rw$. Thus we can obtain a strictly smaller chain, which contradicts the assumption, thus $P_R^N=P_R^{N+1}$.

    $3\rightarrow 2$.
    Finally, let us assume $S_R\subseteq_\finite R\subseteq_\finite  P_R$ and $\exists k\ P_R^k=P_R^{k+1}$.
    Let us assume that $S\subseteq_N  P_R$. We only have to show $S_R\subseteq_{N^k}P_R^k$ to conclude the proof.
    Let us assume that for some $i$ we have $S_R\subseteq_{N^i}P_R^i$. We want to show that $S_R\subseteq_{N^i}P_R^{i+1}$.
    Let $u\in A^+$, let $T= P_R^{i}(u)$. Let $T'\subseteq T$ be such that $\forall v\in T,\ \exists! w\in T',\ vS_Rw$.
    Thus we have $P_R(T)=P_R(T')$ since $S_R$ is the syntactic equivalence relation of $R$.
    Moreover, we have $|T'|\leq N^i$ by assumption, since for all words $v,w\in T'$, $v{{\cancel S_R}}w$.
    For each $v\in T'$, for each $X\subseteq P_R(v)$ verifying $\forall x,y\ x{\cancel{S}_R}y$, we know by assumption that $|X|\leq N$.
    Thus for any $Y\subseteq P_R(T)=P_R(T')$ verifying $\forall x,y\ x{\cancel{S}_R}y$, we know that $|Y|\leq |T'|\cdot N\leq N^{i+1}$, which concludes the proof.
\end{proof}
From this characterization we obtain two decidability results, one negative and one positive.

\begin{theorem}
    \label{thm:undec-seq}
    The following problem is undecidable:

    \textbf{Input:} $\R$ a letter-to-letter transducer realizing an equivalence relation $R$.

    \textbf{Question:} Does $R$ belong to $\kerseq$?
\end{theorem}
The proof of this theorem relies on a reduction of the \emph{mortality problem}, see~\cite[p.~226]{Hooper66} and is given in App.~\ref{app:undec-seq}.
The next theorem shows that we are able to identify exactly where the undecidability comes from: computing the transitive closure of the relation $P_R$.

\begin{theorem}
    The following problem is decidable:

    \textbf{Input:} $\R,\Pp$ two transducers realizing equivalence relations $R,P$, respectively, such that $P$ is the transitive closure of the prefix closure of $R$.

    \textbf{Question:} Does $R$ belong to $\kerseq^\lp$?
\end{theorem}

\begin{proof}
    To show this we rely on the characterization from Thm~.\ref{thm:char-seq2}.
    We proceed as in the proof of Thm.~\ref{thm:dec-ltl}, except that we want to check whether $S_R$ has finite index with respect to $P=P_R^+$ instead of $R$.
    First we can compute a transducer realizing $S_R$, according to Prop.~\ref{prop:comp-synt}. Then from \cite[Thm.~5.1]{Johnson85}, we know we can obtain a transducer realizing a function $f$ whose kernel is $S_R$. Then, using Prop.~\ref{prop:dec-fin-index}, we can decide if $S_R$ has finite index with respect to $P$.
\end{proof}
We sum up the decidability of the problem for different classes of equivalence relations in the table of Fig.~\ref{fig:table}.
New results are shown in red.
\begin{figure}
    \centering
    \begin{normalsize}
\begin{tabular}{|c|c|c|c|}
    \hline
    $\quad$Relations $\backslash$ Kernels$\quad$ &$\quad \kerseq^\ltl\quad$ & $\quad\kerseq\quad$ &$\quad \kerrat\quad$ \\
    \hline
    $\rateq^\ltl$ & {\color{red}\textsf{Dec.}} & {\color{red}\textsf{Undec.}} (Thm.~\ref{thm:undec-seq}) & \textsf{Yes}\\
    \hline
    $\rateq^{{\color{white} \ltl}}$ & {\color{red}\textsf{Dec.}} (Thm.~\ref{thm:dec-ltl}) & {\color{red}\textsf{Undec.}} & \textsf{?}\\
    \hline
   
\end{tabular}
\end{normalsize}
\caption{Summary of the results.}
\label{fig:table}
\end{figure}

\section*{Conclusion}
We have studied the observation synthesis problem for two classes of observation functions: $\kerseq$ and $\kerseq^\ltl$.
A natural question would be to consider the same problem for different classes of functions. 
However, the term \emph{observation function} is only justified (and related to games with imperfect information) if the functions considered are \emph{monotone} meaning that if $h_1\prec h_2$ denotes that history $h_1$ is a prefix of history $h_2$, then any reasonable class of observation function should ensure that $f(h_1)\prec f(h_2)$, for any function $f$.

Since bounded memory and monotonicity somehow characterize the sequential functions, this means that such a class of observation functions would have to use unbounded memory, for instance the class of regular function, \textit{i.e.}~functions realized by two-way transducers. In terms of observations, this would mean that a single game step could give an arbitrary long (actually linear in the size of the history) sequence of observations.

\section*{Acknowledgements}
We would like to thank Bruno Guillon for his help in obtaining the undecidability result.
\bibliographystyle{}
\bibliography{biblio}

\appendix

\section{Proof of Lem.~\ref{lem:fplus-close}}
\label{app:fplus-close}
\begin{proof}
    $R\in \kersseq^\lp$, let $f$ be a subsequential function such that $\ker(f)=R$ and let $P$ be the prefix closure of $R$.
    Let $(\T,\final)$ be a subsequential transducer realizing $f$ and let $K$ be 
    the maximal size of an output of $(\T,\final)$ .
    Let us remark that there are no loops in $\T$ that produce nothing. Indeed, if we assume otherwise, then we can find a loop in $\T$ producing nothing, contradicting the fact that $R$ is length-preserving.
    Hence let $k$ be the smallest ratio of output length over input length for a simple loop in $\T$.
    Thus we have for any words $u, v$, $k|v|-b\leq ||f(uv)|-|f(u)||\leq K|v|$. 

    Let us assume towards a contradiction that the statement does not hold. This means that for any $D$, we can find a sequence $u_0Pu_1P\ldots Pu_N$, such that $||f(u_0)|-|f(u_N)||\geq D$.
    Without loss of generality, let us assume that $|f(u_0)|$ is minimal among $\set{|f(u_i)|\ \mid\ 0\leq i\leq N}$.
   According to Lem.~\ref{lem:fut-out-close} there exists $\delta$ such that $||f(u_{i-1})|-|f(u_{i})||\leq \delta$, for any $i\in\set{1,\ldots,N}$. Thus, for any integer $M\in \set{|f(u_0)|,|f(u_0)|+1,\ldots, |f(u_N)|}$, there exists $i\in\set{0,\ldots,N} $ and $d\in \set{0,\ldots, \delta}$ such that $|f(u_i)|+d=M$.

   Let $C> \max(3,K,\frac {2K} k,\delta,b)$ be a large enough integer.
   We extract a subsequence of the $u_i$s defined in the following way.
   Let $i\geq 0$ be such that $|f(u_0)|+C^i+\delta \leq |f(u_N)|$, then there exists $u\in \set{u_0,\ldots,u_N}$ such that $|f(u)|=C^i+d_i$, with $d_i\in \set{0,\ldots, \delta}$, and we set $v_i=u$.
   Let $v_i'$ be the smallest prefix of $v_i$ such that $|f(v_{i}')|= |f(u_0)|+d_i'$ with $d_i'\in \set{0,\ldots, \delta}$. Then we obtain:

   $$\begin{array}{rcccl}
       k\dist{v_i,v_i'}-b &\leq&  \dist{f(v_{i}),f(v_{i}')} &\leq & K\dist{v_i,v_i'}\\
       k\dist{v_i,v_i'}-b &\leq & C^i+d_i-d_i' &\leq& K\dist{v_i,v_i'}\\
   \end{array}$$
Using these inequalities for $i+1$ and $i$ we have:

$$\begin{array}{rllllll}
    \dist{v_{i+1},v_{i+1}'}
    &\geq & \frac{C}{K}C^i+\frac{d_{i+1}-d_{i+1}'}{K} \\
    &\geq & \frac{C}{K}(k\dist{v_i,v_i'}-b+d_i'-d_i)+\frac{d_{i+1}-d_{i+1}'}{K} \\
    &\geq & C\frac{k}{K}\dist{v_i,v_i'}+\frac{C(-b+d_i'-d_i)+d_{i+1}-d_{i+1}'}{K} \\
    &\geq & C\frac{k}{K}\dist{v_i,v_i'}-\frac{C(b+\delta)+\delta}{K} \\
    &> & C\frac{k}{K}\dist{v_i,v_i'}-\frac{2C^2+C}{K} \\
    
\end{array}$$
It suffices to show $C\frac{k}{K}\dist{v_i,v_i'}-\frac{2C^2+C}{K}\geq \dist{v_i,v_i'}$ in order to obtain $\dist{v_{i+1},v_{i+1}'}> \dist{v_i,v_i'}$.

$$\begin{array}{rrrllll}
  &C\frac{k}{K}\dist{v_i,v_i'}-\frac{2C^2+C)}{K} &\geq &\dist{v_i,v_i'}\\
  \Longleftrightarrow\quad &Ck\dist{v_i,v_i'}-(2C^2+C) &\geq &K\dist{v_i,v_i'}\\
  \Longleftrightarrow\quad & \dist{v_i,v_i'} &\geq &\frac{2C^2+C}{Ck-K}
    
\end{array}$$
Since $C>\frac{2K}{k}$, we only have to show $\dist{v_i,v_i'} \geq 2C^2+C$. Moreover, we know that $\dist{v_i,v_i'} \geq \frac{C^i-\delta}{K} $. Thus it suffices to show that $\frac{C^i-C}{C}\geq 2C^2+C$, since $C$ is larger than both $K$ and $\delta$. The inequality holds, as long as $i\geq 4$, since $C$ is larger than 3.

This means that $\dist{v_{i+1},v_{i+1}'}> \dist{v_i,v_i'}$ for any $i>3$.
Since all $v_i$s have the same length, this means that all $v_i'$s have different length, for $i>3$.
For $D$ large enough, we can assume that there are more than $B^{\delta+1}$ $v_i'$s of different lengths. Thus there must exist two with the same image, which contradicts the assumption that $R$ is length-preserving.
\end{proof}

\section{Proof of Thm.~\ref{thm:undec-seq}}
\label{app:undec-seq}
\begin{proof}
    We use a reduction from the following problem, which well call the \emph{bounded configuration problem}:

    \textbf{Input:} $M$ a reversible Turing machine

    \textbf{Question:} Is there a computation $c_1\rightarrow c_2 \rightarrow \ldots$ which visits an infinite number of configurations.

    We first give the reduction and then show that the problem is actually undecidable.
    
    Let $M$ be a Turing machine with alphabet $\Sigma$, a state space $Q$ and a transition function $\delta:Q\times\Sigma\rightarrow Q\times \Sigma\times \set{\mathrm{left}, \mathrm{right}}$.
    A \emph{configuration} is a word over $\Sigma\cup Q$, with exactly one occurrence of a letter in $Q$.

    We define a letter-to-letter transducer $\R$ recognizing an equivalence relation $R$, with a prefix closure $P$.
    Let $c_1,c_2$ be a pair of consecutive configurations, then $\R$ recognizes the pairs $(c_1\sharp1,c_2\sharp2)$, and $(c_2\sharp2,c_1\sharp1)$ by symmetry. Note that these equivalence classes of $R$ have size $2$. Words of the shape $c\sharp$, with $c$ a configuration are only equivalent to themselves. Words that are strict prefixes of words of the shape $c\sharp$ are all equivalent, if they have the same size. All other words are only equivalent to themselves.

    On can easily see that there exists $k$ such that $P^k=P^{k+1}$ if and only if computations of $M$ visit at most $k+1$ different configurations. We only have left to check that there are computations of unbounded size if and only if there is an infinite computation. Let $c_1,c_2,\ldots$ be configurations sur that from $c_n$, the machine $M$ visits at least $n$ distinct configurations. Then we can extract a subsequence $d_1,d_2,\ldots$ such that all configurations start in the same state. Extracting a subsequence we can assume that all cells of the tape at distance $1$ from the reading head agree. Repeating the operation, we end up with a configuration which visits more than $n$ configurations for any $n$, \textit{i.e.}~is infinite.

    \begin{claim}
        The bounded configuration problem is undecidable.
    \end{claim}
    \begin{proof}
        This is shown by a reduction from the \emph{mortality problem} which amounts to deciding if a Turing machine has an infinite computation.
        Note that this is different from the halting problem, because we ask if the machine halts on \emph{all} possible configurations. This problem was shown to be undecidable in \cite[p.~226]{Hooper66} for Turing machines and in \cite[Thm.~7]{KariO08} for \emph{reversible} Turing machines.

        Assume that the bounded configuration problem is decidable. Given a reversible machine $M$, if is has a computation visiting an infinite number of configurations, then it has an infinite computation. If there is no computation visiting an infinite number of distinct configurations then there is a uniform bound on the number of configurations that a computation can visit. This bound $k$ can be computed, just by simulating the machine on larger and larger configurations. Then, on can easily see if one of the computations loops, and thus decide if the machine has an infinite computation.
    \end{proof}

\end{proof}

\end{document}